\documentclass[letterpaper, 10 pt, conference]{ieeeconf}

\makeatletter
\def\munderbar#1{\underline{\sbox\tw@{$#1$}\dp\tw@\z@\box\tw@}}
\makeatother

\usepackage{tabularx}

\usepackage{algpseudocode}
\usepackage{algorithm}
\usepackage{algorithmicx}
\usepackage{lipsum} 
\usepackage{arydshln}
\usepackage[dvips]{color}
\usepackage{comment}
\usepackage{todonotes}
\usepackage{epsf}
\usepackage{epsfig}
\usepackage{times}
\usepackage{epsfig}
\usepackage{graphicx}
\usepackage{bbold}
\usepackage{mathrsfs}
\usepackage{amssymb}
\usepackage{epstopdf}
\newfloat{algorithm}{t}{lop}

\usepackage{amsmath}
\usepackage{cite}
\usepackage{dsfont}

\usepackage{amsmath,epsfig,amssymb,algorithm,algpseudocode,amsthm,cite,url}
\usepackage{subcaption}
\allowdisplaybreaks
\usepackage{csquotes}

\usepackage{verbatim}
\usepackage[english]{babel}
\usepackage{amsmath,amssymb}

\usepackage[justification=centering]{caption}
\usepackage{verbatim}

\newtheorem{corollary}{Corollary}
\newtheorem{theorem}{\bf Theorem}

\newtheorem{proposition}{\bf Proposition}
\newtheorem{lemma}{\bf Lemma}

\newtheorem{definition}{\bf Definition}

\begin{document}
	%
	\IEEEoverridecommandlockouts
	\title{Generalized Colonel Blotto Game\vspace{-4mm}}
	\author{Aidin Ferdowsi$^{1,*}$, Anibal Sanjab$^{1,*}$, Walid Saad$^{1}$, and Tamer Ba\c{s}ar$^2$ 
		\thanks{This research was supported by the U.S. National Science Foundation under Grants OAC-1541105 and CNS-1446621, and by the Army Research Laboratory under IoBT Cooperative Agreement Number W911NF-17-2-0196.}
		\thanks{$^{1}$Aidin Ferdowsi$^*$, Anibal Sanjab$^*$ and Walid Saad are with Wireless@VT, Bradley Department of Electrical and Computer Engineering, Virginia Tech, Blacksburg, VA, USA,
			{\tt\small \{aidin,anibals,walids\}@vt.edu}}%
		\thanks{$^{2}$Tamer Ba\c{s}ar is with the Coordinated Science Laboratory, University of Illinois at Urbana-Champaign, IL, USA,  
			{\tt\small basar1@illinois.edu}}%
		\thanks{$^{*}$The first two authors have equally contributed to this paper.}
	\vspace*{-100cm}}
\makeatletter
\patchcmd{\@maketitle}
{\addvspace{0.5\baselineskip}\egroup}
{\addvspace{-5\baselineskip}\egroup}
{}
{}
\makeatother
	\maketitle
	\thispagestyle{empty}
	\pagestyle{empty}
	

	%
	\IEEEpeerreviewmaketitle
\begin{abstract}
	Competitive resource allocation between adversarial decision makers arises in a wide spectrum of real-world applications such as in communication systems, cyber-physical systems security, as well as financial and political competition. Hence, developing analytical tools to model and analyze competitive resource allocation is crucial for devising optimal allocation strategies and anticipating the potential outcomes of the competition. To this end, the Colonel Blotto game is one of the most popular game-theoretic frameworks for modeling and analyzing such competitive resource allocation problems. However, in many practical competitive situations, the Colonel Blotto game does not admit solutions in deterministic strategies and, hence, one must rely on analytically complex mixed-strategies with their associated tractability, applicability, and practicality challenges. In this regard, in this paper, a generalization of the Colonel Blotto game which enables the derivation of deterministic, practical, and implementable equilibrium strategies is proposed while accounting for scenarios with heterogeneous battlefields. In addition, the proposed generalized game factors in the consumed/destroyed resources in each battlefield, a feature that is not considered in the classical Blotto game. For this generalized game, the existence of a Nash equilibrium in pure strategies is shown. Then, closed-form analytical expressions of the equilibrium strategies are derived and the outcome of the game is characterized, based on the number of each player's resources and each battlefield's valuation. The generated results provide invaluable insights on the outcome of the competition. For example, the results show that, when both players are fully rational, the more resourceful player can achieve a better total payoff at the Nash equilibrium, a result that is not mimicked in the classical Blotto game. 
\end{abstract}\vspace{-0.3cm}
\section{Introduction}
Game-theoretic resource allocation models have become prevalent across a variety of domains such as wireless communications\cite{Wu2012anti} and \cite{2009Saad}, 
cyber-physical security~\cite{ferdowsi2017game,Gupta2014,ferdowsi2017secure}, financial markets, and political/electoral competitions\cite{Thomas2017,2018Behnezhad,Myerson1993}. 
One of the mostly adopted game-theoretic frameworks for modeling and analyzing such competitive resource allocation problems is the so-called \emph{Colonel Blotto game} (\emph{CBG})~\cite{Roberson2006}. The CBG captures the competitive interactions between two players that seek to allocate resources across a set of battlefields. The player who allocates more resources to a certain battlefield wins it and receives a corresponding valuation. The fundamental question that each player seeks to answer in a CBG is how to allocate its resources to maximize the valuations acquired from the won battlefields. In recent years, many variants of the CBG have been studied including those with homogeneous valuations~\cite{Myerson1993,Sahuguet2006}, symmetric resources, multiple players~\cite{Gupta2014}, and heterogeneous resources~\cite{schwartz2014heterogeneous,Dziubinski2013,Weinstein2012}.

This existing body of literature has primarily focused on analyzing the existence of the equilibrium in a CBG. For instance, the seminal work in~\cite{Roberson2006} has proven that an equilibrium in deterministic strategies (i.e. pure-strategies) of the Blotto game does not exist when the number of considered battlefields is greater than two. As a result, most of the literature has henceforth focused on analyzing the mixed-strategy equilibrium of the CBG, which essentially corresponds to a multi-variate probability distribution over the resources to be allocated over each of the battlefields. In this regard, the works in \cite{Dziubinski2013,Roberson2006,Weinstein2012} have studied uniform marginal distributions of resources on each battlefield. 

However, relying on mixed strategies, as has been the case in all of the CBG literature~\cite{Roberson2006,Dziubinski2013,Weinstein2012,Sahuguet2006,Myerson1993,schwartz2014heterogeneous}, presents serious challenges to the tractability, applicability, and implementability of the derived solutions in real-world environments. In fact, deriving such mixed strategies is often overly complex requiring various mathematical simplifications along the way for tractability~\cite{tan2011fair,2009Saad,Wu2012anti,wu2009optimal,ferdowsi2017game,Gupta2014}. Moreover, in terms of applicability, the optimality of such mixed strategies assumes that the game is actually played for an infinitely large number of times~\cite{bacsar1998dynamic}. In addition, previous CBG models assume win-or-lose settings over each battlefield. In other words, the player that allocates more resources over a battlefield wins all of its valuation and the other player receives nothing. However, given that resources are consumed over a certain battlefield, even when this battlefield is won or lost, these consumed resources must be accounted for in the game model. For instance, even when losing a battlefield, destroying a portion of the resources of the opponent can be considered a gain for the non-winner.   

The main contribution of this paper is therefore to develop a fundamentally novel generalization of the CBG dubbed as the \emph{Generalized Colonel Blotto Game} (\emph{GCBG}) which captures the rich and broad settings of the classical CBG, and, yet enables the derivation of tractable, deterministic, and practical solutions to the two-player competitive resource allocation problem while not being limited to studying win-or-lose settings over each battlefield. To this end, we first introduce the basics of the proposed GCBG and, then, we compare its features to the classical CBG. Subsequently, we prove the existence of an equilibrium to the GCBG -- a pure-strategy Nash equilibrium (NE) -- and provide a detailed derivation as well as closed-form analytical expressions of the pure NE strategies of the GCBG. In addition, using the derived NE strategies, we prove that when both players are fully rational, the more resourceful player is able to achieve an overall payoff that is greater than that of its opponent, and hence, win the game. In contrast with the classical CBG, in our proposed GCBG:
\begin{itemize}
	\item We derive deterministic equilibrium strategies which are practical, in the sense that the players can derive exact optimal strategies and do not need to randomize between possible strategies,
	\item We characterize low complexity solutions even for a large number of battlefields with heterogeneous valuations,
	\item We model realistic applications of competitive resource allocation problems by taking into consideration partial winning and losing over a battlefield.
\end{itemize}

Finally, we provide a numerical example which showcases the effect of the number of resources of each player on the NE strategies and outcomes of the GCBG.

\section{Proposed Generalized Colonel Blotto Game}\label{sec:game model}
\subsection{Classical Colonel Blotto Game}
A standard CBG\cite{Roberson2006} $\hspace{-0.1cm}  \Big\{\hspace{-0.05cm}\mathcal{P},\{\mathcal{Q}^j\}_{j\in \mathcal{P}}, \{R^j\}_{j\in \mathcal{P}},\mathcal{N}, \{v_i\}_{i=1}^n,$ $ \{U^j\}_{j\in\mathcal{P}} \hspace{-0.05cm}\Big\} \hspace{-0.13cm}$ is defined by six components: a) the \emph{players} in the set $ \mathcal{P}\triangleq \{a,b\} $, b) the \emph{strategy} spaces $ \mathcal{Q}^j $ for $ j \in \mathcal{P} $, c) the \emph{available resource} $ R^j $ for $ j \in \mathcal{P} $, d) a set of $n$ \emph{battlefields}, $\mathcal{N}$, e) the normalized \emph{value} of each  battlefield, $ v_i $ for $i\in\mathcal{N}$, and f) the \emph{utility function}, $ U^j $, for each player. 
The set of possible strategies for each player, $\mathcal{Q}^j$, corresponds to the different possible distributions of its $R^j$ resources across the battlefields. Hence, $\mathcal{Q}^j$ is defined as \vspace{-0.1cm}
\begin{align}\label{strategy_set}
\mathcal{Q}^j=\left\{\boldsymbol{r}^j\Bigg|\sum_{i=1}^{n} r_i^j\leq R^j,r_i^j\geq 0  \right\},
\end{align}
where $ r_i^j $ is the number of allocated resource units by player $j$ to battlefield $i$, and $\boldsymbol{r}^j\in\mathcal{Q}^j$ is the vector of these allocated resources (an allocation strategy of player $j$) . To this end, the payoff of player $j$, for a battlefield $i$, is defined as:\vspace{-0.1cm}
\begin{align}\label{payoff}
u^j_i(r^j_i,r^{-j}_i)=
\begin{cases}
v_i, & \textrm{if } r^j_i> r^{-j}_i,\\
\frac{v_i}{2}, & \textrm{if } r^j_i= r^{-j}_i,\\
0, &\textrm{if } r^j_i< r^{-j}_i,
\end{cases}
\end{align}
where $r_i^{-j}$ corresponds to the resources allocated by the opponent of $j$ to battlefield $i$. Such resource allocation strategies are known as pure deterministic strategies since the allocation is not randomized over the battlefields. The utility function of each player, $ U^j $, over all battlefields in $\mathcal{N}$ is defined as:\vspace{-0.2cm}
\begin{align}\label{Totalpayoff}\vspace{-3mm}
U^j(\boldsymbol{r}^j,\boldsymbol{r}^{-j})\hspace{-1mm}=\hspace{-1mm}\sum_{i=1}^{n}u^j_i(r^j_i,r^{-j}_i).
\end{align}

Each player in $\mathcal{P}$ aims to choose a resource allocation strategy that maximizes its payoff over the $n$ battlefields given the potential resource allocation strategy of its opponent. As such, the pure-strategy NE of the CBG can be formally defined as follows:\vspace{-2mm}
\begin{definition}
	A strategy profile $({\boldsymbol{r}^j}^*,{\boldsymbol{r}^{-j}}^*)$, such that ${\boldsymbol{r}^j}^* \in \mathcal{Q}^j$ and ${\boldsymbol{r}^{-j}}^* \in \mathcal{Q}^{-j}$, is a pure-strategy \emph{Nash equilibrium} of the CBG if
	\begin{align}
	U^j({\boldsymbol{r}^j}^*,{\boldsymbol{r}^{-j}}^*) \geq 	U^j({\boldsymbol{r}^j},{\boldsymbol{r}^{-j}}^*), \forall {\boldsymbol{r}^j} \in \mathcal{Q}^j.
	\end{align}
\end{definition}

For the CBG, it has been proven~\cite{Roberson2006} that for $R^b>R^a $ and $nR^a>R^b$, there exists no pure-strategy NE. As a result, for solving the CBG, a common methodology is to explore mixed strategies based on which each player $j\in\mathcal{P}$ chooses a probability distribution (or a multi-variate probability density function) over $\mathcal{Q}^j$ to maximize its potential expected payoff. However, the reliance on such mixed strategies introduces high analytical complexity and presents challenges in terms of the tractability of the derivation of the solutions as well as to the practicality and applicability of these solutions. Hence, to enable the derivation of tractable, deterministic, and practical solutions to the competitive resource allocation problem, we next propose a generalization of the CBG dubbed generalized Colonel Blotto game. 

\subsection{Proposed Generalized Colonel Blotto Game}\label{sec:GCB}
A GCBG $ \Big\{\mathcal{P},\{\mathcal{Q}^j\}_{j\in \mathcal{P}}, \{R^j\}_{j\in \mathcal{P}},\mathcal{N}, \{v_i\}_{i=1}^n, \{\tilde{U}^j\}_{j\in\mathcal{P}},$ $k\Big\} $ is defined by seven components that include the players, strategy spaces, available resources, battlefields, and valuation of battlefields as is the case in the CBG. However, we define a new utility function $ \{\tilde{U}^j\}_{j\in\mathcal{P}} $ as an approximation to the utility function of the CBG. This approximation is based on a continuous differentiable function and an approximation parameter, $k$. When $k$ increases, the GCBG will very closely resemble the CBG. At the limit, when $k$ goes to infinity, the GCBG converges to the CBG. However, in contrast with the CBG, the differentiability of the approximate utility function enables the derivation of deterministic equilibrium strategies. 

In this respect, as defined in~(\ref{payoff}), the payoff from each battlefield resembles a step function. Hence, we propose an approximation of this function, $\tilde{u}^j_i(r^j_i,r^{-j}_i,k)$, referred to hereinafter as $k-$approximation, defined as follows:
\begin{align}\label{kpayoff}
\tilde{u}^j_i(r^j_i,r^{-j}_i,k)\triangleq\frac{v_i}{\pi}\arctan\left(k(r^j_i-r^{-j}_i)\right)+\frac{v_i}{2},
\end{align}
and the utility function of each player $ \tilde{U}^j $ is the summation of the $k-$approximation payoffs from each battlefield. Next, in Lemma~\ref{atan}, we will show that \eqref{kpayoff} precisely approximates $u_i^j(r_i^j,r_i^{-j})$ as the approximation factor $k$ tends to infinity. Fig.~\ref{fig:atan} provides a number of plots of $k-$approximation functions for different values of $k$.

\begin{lemma}\label{atan}
	The payoff from each battlefield $i\in\mathcal{N}$ can be expressed as follows:
	\begin{align}
	u^j_i(r^j_i,r^{-j}_i)=\lim\limits_{k\rightarrow \infty}\tilde{u}^j_i(r^j_i,r^{-j}_i,k).
	\end{align}
\end{lemma}
\begin{proof}
	The limit of the $k$-approximate payoff when $k\rightarrow\infty $ can be written as:
	\begin{align}\nonumber
	\lim\limits_{k\rightarrow \infty}\tilde{u}^j_i(r^j_i,r^{-j}_i,k)\hspace{-1mm}=\hspace{-1mm}\begin{cases}
	\frac{v_i}{\pi}\arctan\left(+\infty\right)+\frac{v_i}{2}\hspace{-0.5mm}=\hspace{-0.5mm}v^j_i,&\hspace{-2mm}r_i^j \hspace{-0.5mm} >\hspace{-0.5mm} r_i^{-j}, \\
	0+\frac{v_i}{2}=\frac{v_i}{2},&\hspace{-2mm}r_i^j\hspace{-0.5mm} = \hspace{-0.5mm}r_i^{-j},\\
	\frac{v_i}{\pi}\arctan\left(-\infty\right)+\frac{v_i}{2}\hspace{-0.5mm}=\hspace{-0.5mm}0,&\hspace{-2mm}r_i^j\hspace{-0.5mm} < \hspace{-0.5mm}r_i^{-j}.
	\end{cases}
	\end{align}
	Therefore, as compared to~(\ref{payoff}), 
	\begin{align}
	\lim\limits_{k\rightarrow \infty}\tilde{u}^j_i(r^j_i,r^{-j}_i,k)=u^j_i(r^j_i,r^{-j}_i),
	\end{align}
	which proves the lemma.
\end{proof}

Hence, Lemma \ref{atan} provides an approximation for the payoff from each battlefield. Using $ \tilde{u}^j_i(r^j_i,r^{-j}_i,k) $ instead of $ u^j_i(r^j_i,r^{-j}_i) $ provides two additional benefits as compared to the CBG: a) the continuity and differentiability of $ \tilde{u}^j_i(r^j_i,r^{-j}_i,k) $ helps in deriving pure-strategy NEs, and b) $\tilde{u}^j_i(r^j_i,r^{-j}_i,k)$ captures the notion of partial-win-or-lose by accounting for resource consumption and losses which is a practical feature in a GCBG that has not been considered in the CBG. 
In fact, in practice, the payoff from each battlefield must consider the portion of resources destroyed/consumed following the ``battle'' between the two players. This means that the payoff from each battlefield has to capture the loss due to the resources destroyed in this battlefield, which also corresponds to a gain for the opponent. The $k-$approximate utility function in~(\ref{kpayoff}) inherently captures this aspect. Fig. \ref{atan} shows the continuous transition from losing to winning a certain battlefield when using $k$-approximate payoff functions. Here, we note that the larger $k$ is, the closer the $k$-approximation approximates the classical CBG.

\begin{figure}[!t]
	\centering
	\captionsetup{singlelinecheck = false, justification=justified}\vspace{-0.2cm}
	\includegraphics[width=0.96\columnwidth]{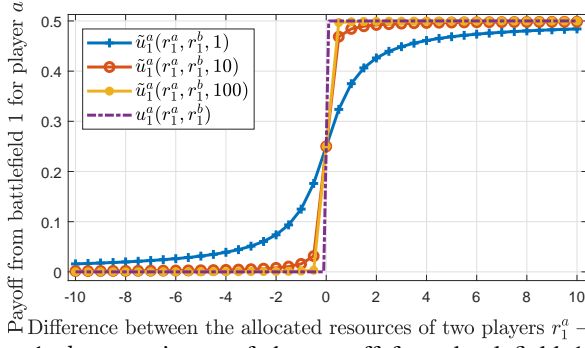}\vspace{-0.3cm}
	\caption{$ k $-approximate of the payoff from battlefield 1 for player $ a $ when $ v_1^a=0.5 $.}
	\label{fig:atan}
	\vspace{-0.7cm}
\end{figure}

Now, we define the $k$-approximate utility function per player:\vspace{-5mm}
\begin{align}\label{kutility}
\tilde{U}^j(\boldsymbol{r}^j,\boldsymbol{r}^{-j},k)\triangleq\sum_{i=1}^{n}\tilde{u}^j_i(r^j_i,r^{-j}_i,k).
\end{align}\vspace{-3mm}

From \eqref{kpayoff} and \eqref{kutility}, and given that $\arctan$ is an odd function, the players' payoffs from each battlefield can be related following $
\tilde{u}^j_i(r^j_i,r^{-j}_i,k)+\tilde{u}^{-j}_i(r^{-j}_i,r^j_i,k)=v_i, $
which results in the following relation between the players' utility functions: \vspace{-0.1cm}
\begin{align}
\label{constantsum}
\tilde{U}^j(\boldsymbol{r}^j,\boldsymbol{r}^{-j},k)+\tilde{U}^{-j}(\boldsymbol{r}^{-j},\boldsymbol{r}^j,k)&=1.
\end{align}

The relation in~\eqref{constantsum} shows that the GCBG is a \emph{constant-sum game} in which the pure-strategy NE is the solution of the following \emph{minimax} problem\cite{bacsar1998dynamic}:\vspace{-0.1cm}
\begin{align}\label{maximin}
V^j(k)\triangleq 1- V^{-j}(k)\triangleq\hspace{-3mm}\min_{\boldsymbol{r}^{-j}\in \mathcal{Q}^{-j}}\max_{\boldsymbol{r}^j\in \mathcal{Q}^j}\tilde{U}^j(\boldsymbol{r}^j,\boldsymbol{r}^{-j},k),
\end{align}
where $ V^j(k) $ is the value of the GCBG with the $k$-approximate utility function for player $j$. 
Next, we solve the minimax problem in \eqref{maximin} and derive the pure-strategy  NE of the GCBG.
\vspace{-3mm}
\section{Solution of the GCBG}\label{sec:solution}

To solve the minimax problem in \eqref{maximin}, we must derive $\boldsymbol{r}^{a*}\in\mathcal{Q}^a$ which maximizes $\tilde{U}^a(\boldsymbol{r}^a,\boldsymbol{r}^{b},k)$ for all possible $\boldsymbol{r}^{b}\in\mathcal{Q}^b$. Then, we choose $\boldsymbol{r}^{b*}\in\mathcal{Q}^b$ which minimizes  $\tilde{U}^a(\boldsymbol{r}^{a*},\boldsymbol{r}^{b},k)$ considering the response, $\boldsymbol{r}^{a*}$, to any chosen $\boldsymbol{r}^{b}\in\mathcal{Q}^b$. Given the resource limitation of each player, a limit exists on the chosen allocations as expressed in~(\ref{strategy_set}). First, we prove that this limit is binding, i.e. each player is always better off spending all their available resources. \vspace{-0.1cm}

\begin{proposition}\label{lemmaresources}
	All the strategies of the form $\sum_{i=1}^{n}r_i^j<R^j $ are dominated by the strategies  $\sum_{i=1}^{n}r_i^j=R^j $ i.e, player $ j $ is better off using all of its available resource.
\end{proposition}
\begin{proof}
	The proof is by contradiction. Suppose $ \bar{\boldsymbol{r}}^j=[\bar{r}^j_1,\dots,\bar{r}^j_n] $ is the vector that maximizes $ \tilde{U}^j(\boldsymbol{r}^j,\boldsymbol{r}^{-j},k)$ such that: $	\sum_{i=1}^n \bar{r}^j_i =\bar{R}^j<R^j.	$
	Then, the payoff for player $ j $ choosing strategy $ \bar{\boldsymbol{r}}^j $ is:\vspace{-2mm}
	\begin{align}
		\tilde{U}^j(\bar{\boldsymbol{r}}^j,\bar{\boldsymbol{r}}^{-j},k)= \sum_{i=1}^{n}\frac{v_i}{\pi}\arctan\left(k(\bar{r}_i^j-\bar{r}_i^{-j})\right) +\frac{1}{2}.\nonumber
	\end{align}
	
	Now, considering arbitrarily player $a$, we define a new allocation vector for player $ a $ as $ \boldsymbol{\rho}^a=[\bar{r}^a_1+R^a-\bar{R}^a,\bar{r}^a_2,\dots,\bar{r}^a_n] $. Then, the payoff for player $ a $ will be:
	\begin{align}
		&\tilde{U}^a(\boldsymbol{\rho}^a,\bar{\boldsymbol{r}}^b,k)=\frac{v_1}{\pi}\arctan\left(k(\bar{r}^a_1+R^a-\bar{R}^a-\bar{r}_i^b)\right)\nonumber\\
		&+\sum_{i=2}^{n}\frac{v_i}{\pi}\arctan\left(k(\bar{r}_i^a-\bar{r}_i^b)\right) +\frac{1}{2}\nonumber\\
		&>\frac{v_1}{\pi}\arctan\left(k(\bar{r}_1^a-\bar{r}_1^b)\right)+\sum_{i=2}^{n}\frac{v_i}{\pi}\arctan\left(k(\bar{r}_i^a-\bar{r}_i^b)\right)+\frac{1}{2}\nonumber\\\nonumber
		&=\tilde{U}^a(\bar{\boldsymbol{r}}^a,\bar{\boldsymbol{r}}^b,k),
	\end{align}
which contradicts with the assumption $ \bar{\boldsymbol{r}}^j=[\bar{r}^j_1,\dots,\bar{r}^j_n] $  maximizes $ \tilde{U}^j(\boldsymbol{r}^j,\boldsymbol{r}^{-j},k)$. 
\end{proof}

From Proposition \ref{lemmaresources}, we know that the limit on the total number of allocated resources is binding. Before characterizing the solution of the game, we first define the difference between the allocated resources by each player on each battlefield $i\in\mathcal{N}$, $z_i$, and the difference between the total available resources, $D$, as follows:\vspace{-3mm}
\begin{align}\label{zi}
z_i=r_i^a-r_i^b,\,\,\,D=R^a-R^b.
\end{align}

In addition, we consider, without loss of generality, that $R^b<R^a$, which implies that the available resources of an arbitrary player $b$ are less than those of its opponent, player $a$. We also consider the case in which $R^a/n<R^b$, which eliminates the trivial case studied in the CBG, in which, if $R^a>nR^b$, player $a$ can trivially win the game. Using these two inequality assumptions on $R^a$ and $R^b$, the constraint on $D$ can be expressed as follows:\vspace{-0.2cm}
\begin{align}\label{Dlimit}
0<D<(n-1)R^b.
\end{align}

From Proposition \ref{lemmaresources} and \eqref{zi}, we can write:\vspace{-0.2cm}
\begin{align}\label{sumdif}
\sum_{i=1}^{n} z_i=D.
\end{align}\vspace{-0.2cm}

Since the utility functions for both players are dependent on the difference of allocated resources on each battlefield, $ z_i $, hereinafter, we use $\tilde{U}^a(\boldsymbol{z},k)$ and $\tilde{U}^a(\boldsymbol{r}^a,\boldsymbol{r}^b,k)$ interchangeably, where $ \boldsymbol{z}=[z_1,\dots,z_n]$. In what follows, we fully characterize the pure-strategy NE for the GCB game. 

First, we focus on solving the maximization component of the minimax problem in~(\ref{maximin}). In this regard, Theorem~\ref{theoremlocamax} characterizes the local maxima of $\tilde{U}^a(\boldsymbol{r}^a,\boldsymbol{r}^b,k)$. Here, without loss of generality, we assume that the battlefields are indexed based on an increasing order of their value, i.e. $v_1\geq v_2 \geq \dots \geq v_n$ with one of these inequalities being strict to avoid solving a trivial case in which all battlefields are identical. \vspace{-0.2cm}  

\begin{theorem}\label{theoremlocamax}
	A local maximum $\boldsymbol{z}^*=[z_1^*,\dots,z_n^*]$ of $\tilde{U}^a(\boldsymbol{z},k)$ with respect to $\boldsymbol{z}$, is a solution to the following equation:\vspace{-0.2cm}
	\begin{align}\label{DConstraint}
	\sum_{i=1}^{n-1}z_i^*+z_n=D,
	\end{align}
	where\vspace{-0.2cm}
	\begin{align}\label{solution1tonminus1}
	z^*_i=\sqrt{\frac{1}{k^2v_n}\left(k^2z^2_nv_i+v_i-v_n\right)},\,\, \textrm{for } i\in\mathcal{M},
	\end{align}
	and $ z_n>0. $\vspace{-0.2cm}
\end{theorem}
\begin{proof}
	Based on Proposition~\ref{lemmaresources}, $z_n$ can be expressed as $z_n=D-\sum_{i=1}^{n-1} z_i$, and player $a$'s utility function will be:
	\begin{align}
		\tilde{U}^a(\boldsymbol{z},k)&=\sum_{i=1}^{n-1}\frac{v_i}{\pi}\arctan(kz_i)\nonumber\\
		&+\frac{v_n}{\pi}\arctan\left(k(D-\sum_{i=1}^{n-1} z_i)\right)+\sum_{i=1}^{n}\frac{v_i}{2}.
	\end{align}
	
	Hence, to find the local maxima of player $ a $'s utility function, from the first order necessary condition of optimality, we know that, if $ \boldsymbol{z}^* $ is a local maximizer and $ \tilde{U}^a $ is continuously differentiable, then
$\nabla \tilde{U}^a(\boldsymbol{z}^*,k)=\boldsymbol{0}.$ Therefore, for $ i\in \mathcal{M} \triangleq \mathcal{N}\setminus \{n\} $, we must have $ \frac{\partial \tilde{U}^a(\boldsymbol{z},k)}{\partial z_i}=0, $ and, hence:
	\begin{align}\label{FirstDerivative}
		\frac{1}{\pi}\frac{kv_i}{k^2z_i^2+1}-\frac{1}{\pi}\frac{kv_n}{k^2(D-\sum_{l=1}^{n-1} z_l)^2+1}&=0,\nonumber\\
		\frac{kv_i}{k^2z_i^2+1}-\frac{kv_n}{k^2z_n^2+1}&=0.
	\end{align}
	
	Thus, $z_i^*$ must be one of two solutions:
	\begin{align}\label{solutions}
		z_i^*=\pm\sqrt{\frac{1}{k^2v_n}\left(k^2z^2_nv_i+v_i-v_n\right)},\,\,i\in \mathcal{M},
	\end{align}
	which, from \eqref{sumdif}, must meet $ \sum_{i=1}^{n-1}z_i^*+z_n=D $. 
	
	Next, from the second order necessary conditions of optimality, we know that, if $ \boldsymbol{z}^* $ is a local maximizer of $ \tilde{U}^a(\boldsymbol{z},k)$, then the Hessian matrix, $ \nabla^2\tilde{U}^a(\boldsymbol{z}^*,k)$, must be negative definite. The second order derivatives, and elements of the Hessian matrix, are calculated as follows:
	 \begin{align}
	 	\frac{\partial^2\tilde{U}^a(\boldsymbol{z},k)}{\partial z_i^2}&\text{$=$} \text{$-$}\frac{1}{\pi}\frac{2k^3v_iz_i}{(k^2z^2_i\text{$+$}1)^2}\textrm{$-$}\frac{1}{\pi}\frac{2k^3v_n(D\text{$-$}\sum_{l=1}^{n\text{$-$}1} z_l)}{(k^2(D\text{$-$}\sum_{l=1}^{n\text{$-$}1} z_l)^2\text{$+$}1)^2}\nonumber\\
	 	&\text{$=$}\text{$-$}\frac{1}{\pi}\frac{2k^3v_iz_i}{(k^2z^2_i+1)^2}-\frac{1}{\pi}\frac{2k^3v_nz_n}{(k^2z^2_n+1)^2},\\
	 	\frac{\partial^2\tilde{U}^a(\boldsymbol{z},k)}{\partial z_i\partial z_m}&=-\frac{1}{\pi}\frac{2k^3v_n(D-\sum_{l=1}^{n-1} z_l)}{(k^2(D-\sum_{l=1}^{n-1} z_l)^2+1)^2}\nonumber\\
	 	&=-\frac{1}{\pi}\frac{2k^3v_nz_n}{(k^2z^2_n+1)^2}.
	 \end{align} 
	 
	 From \eqref{FirstDerivative}, we have:
	 \begin{align}\label{FirstDerivative2}
	 	\frac{k^2z_n^2+1}{k^2z_i^{*^2}+1}=\frac{v_n}{v_i}.
	 \end{align}
	 
	 Plugging~(\ref{FirstDerivative2}) in the expressions of the elements of the Hessian matrix yields:
	 \begin{align}
	 	\frac{\partial^2\tilde{U}^a(\boldsymbol{z},k)}{\partial z_i^2}\Bigr|_{\boldsymbol{z}=\boldsymbol{z^*}}&=-\frac{2k^3v_n}{\pi(k^2z^2_n+1)^2}(\frac{v_n}{v_i}z_i^*+z_n),\\
	 	\frac{\partial^2\tilde{U}^a(\boldsymbol{z},k)}{\partial z_i\partial z_m}\Bigr|_{\boldsymbol{z}=\boldsymbol{z^*}}&=-\frac{2k^3v_n}{\pi(k^2z^2_n+1)^2}z_n.
	 \end{align}
	 
	 Then, the Hessian matrix at $\boldsymbol{z}^*$ is expressed as:
	 \begin{align}
	 	\boldsymbol{H}=\nabla^2 \tilde{U}^a(\boldsymbol{z}^*,k)= -\frac{2k^3v_n}{\pi(k^2z^2_n+1)^2}\boldsymbol{H}_1,
	 \end{align}
	 where
	 \begin{align}\label{H1}
	 	\boldsymbol{H}_1\textrm{=}\left[\hspace{-0.5mm}\begin{array}{c c c c}
	 	\frac{v_n}{v_1}z^*_1\textrm{$+$}z_n & z_n & \cdots &z_n \\
	 	z_n & \frac{v_n}{v_2}z^*_2\textrm{$+$}z_n & \cdots & \vdots \\
	 	\vdots & \vdots & \ddots & \vdots \\
	 	z_n & z_n &\cdots & \frac{v_n}{v_{n-1}}z^*_{n-1}\textrm{$+$}z_n
	 	\end{array}\hspace{-0.5mm}\right].
	 \end{align}

For $\boldsymbol{z}^*$ to be a local maximum, $\boldsymbol{H}$ must be negative definite and, hence, $\boldsymbol{H}_1$ must be positive definite. Performing row operations on $\boldsymbol{H}_1$, we obtain an upper-triangular matrix: 
 \begin{align}\label{H2}
	 	\boldsymbol{H}_2=\left[\begin{array}{c c c c c}
	 	z^*_1 & 0 & \cdots & 0 & -z^*_{n-1} \\
	 	0 & z^*_2 & \ddots &\vdots & \vdots \\
	 	\vdots & \ddots & \ddots & 0 & \vdots \\
	 	\vdots &  & \ddots & z^*_{n-2} & -z_{n-1}^*\\
	 	0 & \cdots & \cdots & 0 & \xi
	 	\end{array}\right],
	 \end{align}
	 where 
	 \begin{align}\label{Xi}
	 \xi=\frac{v_n}{v_{n-1}}z_{n-1}^*\left(1+\sum_{i=1}^{n-1}\frac{v_iz_n}{v_nz_i^*}\right). 
	 \end{align}
	 
	$\boldsymbol{H}_1$ is positive definite if and only if its pivots (i.e. its diagonal elements) are all positive since the number of positive pivots of $\boldsymbol{H}_2$ is equal to the number of positive eigenvalues of $\boldsymbol{H}_1$~\cite{nonlinear2013programming}. Hence, by inspecting the diagonal elements of $\boldsymbol{H}_2$ in~(\ref{H2}), for $\boldsymbol{H}_1$ to be positive definite, $z_i^*$ must be positive for $i\in\mathcal{M}\setminus \{n-1\}$. Combining this result with~(\ref{solutions}), we obtain the expressions in~(\ref{solution1tonminus1}) for $i\in\mathcal{M}\setminus \{n-1\}$. 
	
Next, we investigate the sign of $\xi$ in~(\ref{Xi}) which must also be positive for $\boldsymbol{H}_1$ to be positive definite. 
Knowing that $z_i>0$ for $i\in\mathcal{M}\setminus \{n-1\}$, we investigate four different cases regarding the signs of $z_{n-1}^*$ and $z_n$ and their effect on the sign of $\xi$. We prove that $z_n>0$ and $z_{n-1}^*>0$ as defined in~(\ref{solution1tonminus1}) is the only case that leads to a positive definite $\boldsymbol{H}_1$ and, as a result, a negative definite $\boldsymbol{H}$. 

    \textit{Case 1.} {$z_{n-1}^*<0$, and $z_n<0$: If $z_{n-1}<0$ and $z_n<0$, then the $(n-1,n-1)$ element of $\boldsymbol{H}_1$ in~(\ref{H1}) is negative. In addition, for a matrix to be positive definite, none of its diagonal elements can be negative~\cite{nonlinear2013programming}. As such, if $z_{n-1}<0$, and $z_n<0$, then $\boldsymbol{H}_1$ is not positive definite and hence $\boldsymbol{H}$ is not negative definite.}
    
    \textit{Case 2.} {$z_{n-1}^*<0$, and $z_n>0$: We rewrite $\xi$ in~(\ref{Xi}) as $\xi=\frac{v_n}{v_{n-1}}z_{n-1}^*\phi$, where for $z^*_{n-1}<0$ and $z_n>0$,
    \begin{align}
    \phi=1&+\sum_{i=1}^{n-2}\frac{v_iz_n}{v_n\sqrt{z_n^2\frac{v_i}{v_n}+\frac{v_i-v_n}{k^2v_n}}}\nonumber\\
    &- \frac{v_{n-1}z_n}{v_n\sqrt{z_n^2\frac{v_{n-1}}{v_n}+\frac{v_{n-1}-v_n}{k^2v_n}}}\nonumber\\
    =1&+\sum_{i=1}^{n-2}\frac{v_iz_n}{v_nz_n\sqrt{\frac{v_i}{v_n}}\sqrt{1+\frac{v_i-v_n}{k^2v_iz_n^2}}}\nonumber\\
    &-\frac{v_{n-1}z_n}{v_nz_n\sqrt{\frac{v_{n-1}}{v_n}}\sqrt{1+\frac{v_{n-1}-v_n}{k^2v_{n-1}z_n^2}}}\cdot
    \end{align} 
    
    Based on the binomial approximation, $(1+\epsilon)^\alpha\approx1+\alpha\epsilon$ if $|\epsilon|<1$ and $|\alpha\epsilon|<<1$, let $\epsilon_i=(v_i-v_n)/(k^2v_iz_n^2)$. For practical $k$-approximations, i.e., when $k$ is a large number ($k>>0$), $\epsilon_i$ meets the conditions of the binomial approximation and is such that $\epsilon_i\approx\epsilon_j$ for all $i,j\in\mathcal{M}$. As such, let $\epsilon_i=\epsilon$ for all $i\in\mathcal{M}$, $\phi$ can be approximated as:
    \begin{align}
    \phi\approx 1+\sum_{i=1}^{n-2}\frac{\sqrt{v_i}}{\sqrt{v_n}(1+\epsilon/2)}-\frac{\sqrt{v_{n-1}}}{\sqrt{v_n}(1+\epsilon/2)}.
    \end{align}
    
    Hence, since $v_i>v_{n-1}$ for $i<n-1$ we obtain $\phi>0\Rightarrow\xi<0\Rightarrow\boldsymbol{H_1}$ is not positive definite $\Leftrightarrow\boldsymbol{H}$ is not negative definite; for any practically large approximation index $k$.} 
    
    \textit{Case 3.} {$z_{n-1}^*>0$, and $z_n<0$: We express $\xi$ as $\xi=\frac{v_n}{v_{n-1}}z_{n-1}^*\Xi$, where  
    \begin{align}
    \Xi= 1+\sum_{i=1}^{n-1}\frac{v_i}{v_n}\frac{z_n}{\sqrt{z_n^2v_i/v_n+(v_i-v_n)/(k^2v_n)}}.
    \end{align}
    
    Since $z_{n-1}^*>0$, the sign of $\xi$ is the same as the sign of $\Xi$. By taking the first derivative of $\Xi$ with respect to $k$, we can see that $\frac{\partial \Xi}{\partial k}<0$, since $z_n<0$ and $z_i^*>0$ for all $i\in \mathcal{M}$ and $v_i>v_n$ for all $i \in \mathcal{M}$. Hence $\Xi$ is decreasing in $k$. Knowing that $k>0$, the upper bound of $\Xi$, $\bar{\Xi}$, is then
    \begin{align}
    \bar{\Xi}=\lim\limits_{k\rightarrow 0}\Xi=1+k\sum_{i=1}^{n-1}\frac{v_iz_n}{\sqrt{v_n(v_i-v_n)}}\cdot
    \end{align}
    
    Hence, $\xi$ is negative when
    \begin{align}\label{ConditionXi}
    K>\frac{-1}{\sum_{i=1}^{n-1}\frac{v_iz_n}{\sqrt{v_n(v_i-v_n)}}}\cdot
    \end{align}
    
    Here, we note that $v_n$ has the same order of magnitude as $\sqrt{v_n(v_i-v_n)}$. Hence, for a practical approximation index, $k>>0$, the condition in~(\ref{ConditionXi}) always holds true. Thus, this implies that $\boldsymbol{H_1}$ is not positive definite and, hence, $\boldsymbol{H}$ is not negative definite for any practically large $k$.} 
    
    \textit{Case 4.} {$z_{n-1}^*>0$, and $z_n>0$: For this case, $\xi>0$, hence, $\boldsymbol{H}_1$ is positive definite and as a result $\boldsymbol{H}$ is negative definite. Hence, only for $z_{n-1}^*>0$, and $z_n>0$, $\boldsymbol{H}$ is positive definite.}

Therefore, the solutions $ \boldsymbol{z}^* $ proposed in \eqref{solution1tonminus1} and meeting~(\ref{DConstraint}) such that $z_n>0$, are the only possible local maxima of $\tilde{U}^a(\boldsymbol{z},k)$.	
\end{proof}

In Theorem~\ref{theoremlocamax}, we identified the local maxima of $\tilde{U}^a(\boldsymbol{z},k)$ as given in~(\ref{solution1tonminus1}) and meeting $ z_n>0 $. We also proved that these solutions must satisfy~(\ref{DConstraint}). Next, we study the equality constraint in~(\ref{DConstraint}) which must be met. To this end, we let \vspace{-0.1cm}
\begin{align}
f_k(z_n)\triangleq z_n +\sum_{i=1}^{n-1}\sqrt{\frac{1}{k^2v_n}(k^2z_n^2v_i+v_i-v_n)}=D.
\end{align} 

$f_k(z_n)$ is a strictly convex function with a negative minimizing $z_n$ as stated in Proposition~\ref{convexity}. \vspace{-0.2cm}

\begin{proposition}\label{convexity}
	$ f_k(z_n) $ is a strictly convex function whose minimum occurs at a negative $z_n$.
\end{proposition}
\begin{proof}
	To prove the convexity of $ f_k(z_n) $, we compute its second derivative with respect to $z_n$. 
	\begin{align}\label{secondderiv}
		\frac{d^2f}{dz^2_n}&=\sum_{i=1}^{n-1}\hspace{-1mm}\frac{\frac{1}{k^2v_n}\left(k^2z^2_nv_i+v_i-v_n\right) -\frac{v_i}{v_n}z_n^2 }{[\frac{1}{k^2v_n}\left(k^2z^2_nv_i+v_i-v_n\right)]^{3/2}}.
	\end{align} 
	
	By inspecting the numerator of each term in $\frac{d^2f}{dz^2_n}$, every numerator can be reduced to $(v_i-v_n)(k^2v_n)$ which is positive. Hence, $\frac{d^2f}{dz^2_n}$ is a summation of positive terms (at least one of which is strictly positive since at least one $v_i>v_n$ for some $i\in\mathcal{M}$). Hence, $\frac{d^2f}{dz^2_n}>0$ and $f_k(z_n)$ is strictly convex. To obtain the sign of $z_n^{\textrm{min}}$ which minimizes $f_k(z_n)$, we inspect the first derivative of $f_k(z_n)$. Since $z_n^{\textrm{min}}$ is a minimum of $f_k(z_n)$, we must have: 
	\begin{align}\label{Firstderivavtivef}
	\frac{df}{dz_n}\Bigr|_{z_n=z_n^{\textrm{min}}}=1+\sum_{i=1}^{n-1}\frac{\frac{v_i}{v_n}z_n^{\textrm{min}}}{\sqrt{\frac{1}{k^2v_n}\left(k^2z_n^{\textrm{min}^2}v_i+v_i-v_n\right)}}=0.
	\end{align}
	
	Since the denominator of the summation term in~(\ref{Firstderivavtivef}) is always positive, and since $v_i/v_n$ for all $i \in \mathcal{M}$ is always positive, $z_n^{min}$ must be negative to solve~(\ref{Firstderivavtivef}). 
\end{proof}	

Proposition~\ref{convexity} is valuable for characterizing the possible solutions in $z_n$ to $f_k(z_n)=D$, and equivalently to~(\ref{DConstraint}), as stated in Corollary~\ref{corrollary_intersection}. To this end, let $\munderbar{D}\triangleq f_k(z^\text{min}_n)$ be the minimum value of $ f_k(z_n)$ which is unique given that $ f_k(z_n)$ was proven to be strictly convex. \vspace{-0.1cm}

\begin{corollary}\label{corrollary_intersection}
	$ f_k(z_n)=D $ has:
	a) two solutions if $ D>\munderbar{D} $, b) one solution if $D=\munderbar{D} $, or c) no solution if $ D<\munderbar{D} $.
\end{corollary}
\begin{proof}
	From Proposition~\ref{convexity}, we know that $f_k(z_n)$ is a strictly convex function. Therefore any line $D$ parallel to the $z_n$ axis will intersect $f_k(z_n)$ in two points if $D>\munderbar{D}$, one point if $D=\munderbar{D}$, and will not intersect $f_k(z_n)$ if $D<\munderbar{D}$.
\end{proof}

Based on Corollary \ref{corrollary_intersection} and Theorem~\ref{theoremlocamax}, we can conclude that for $ D\geq\munderbar{D} $, we have one or two possible local maxima for $ \tilde{U}^a(\boldsymbol{z},k) $. However, we must check whether these maxima are feasible as they should correspond to $z_n^*>0$, which is a necessary condition as proven in Theorem~\ref{theoremlocamax}. These maxima also must satisfy the following feasibility conditions on $z_i$ for $i\in \mathcal{M}$:\vspace{-0.2cm}
\begin{align}
r_i^a\leq R^a,r_i^b\leq R^b,\Rightarrow -R^b\leq z_i=r_i^a-r_i^b\leq R^a.\label{zconstraint}
\end{align}

Therefore, for any solution $z_n^*>0$ of $ f(z_n)=D $, we must check that $z_i^*$, for $i\in\mathcal{N}$, satisfy \vspace{-0.1cm} 
\begin{align}\label{Maximaconstraints}
0<z_i^*\leq R^a.
\end{align}
\begin{figure}[!t]
	\begin{center}
		\includegraphics[width=0.96\columnwidth]{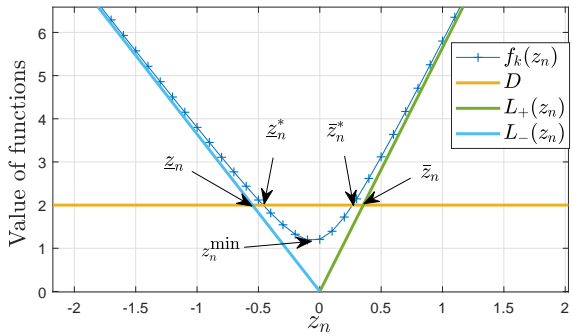}
		\captionsetup{singlelinecheck = false, justification=justified}\vspace{-0.2cm}
		\caption{Comparison of solutions of equations $ f_k(z_n)=D,L_+(z_n)=D$, and $L_-(z_n)=D $.}
		\label{fig:fzn}
	\end{center}\vspace{-0.95cm}
\end{figure}
Next, we show that only one solution of $f_k(z_n)=D$ satisfies~\eqref{Maximaconstraints}.\vspace{-0.2cm}
\begin{theorem}\label{theoremaccsolution}
	\hspace{-1mm}For $ Dk\hspace{-0.5mm}\geq\hspace{-0.5mm}\max\left\{\frac{n-1}{\sqrt{v_n(2n-1)}},\sum_{i=1}^{n-1}\sqrt{{\frac{v_i-v_n}{v_n}}}\right\} $, $ \tilde{U}^a(\boldsymbol{z},k) $ has a unique maximum.
\end{theorem}
\begin{proof}
	From Corollary \ref{corrollary_intersection}, we know that, for $ D>\munderbar{D} $, $ f_k(z_n) =D$ has two solutions which we denote as $ \munderbar{z}_n^*<\bar{z}_n^* $. We need to prove that only one of these solutions satisfies \eqref{Maximaconstraints}. First, we prove that $ f_k(z_n) $ has two slant asymptotes $ L_1(z_n) $ and $L_2(z_n)$, where $ f_k(z_n)>L_+(z_n),L_-(z_n) $. To prove this statement, we find the linear approximations of $ f_k(z_n) $ as $z_n$ goes to $+\infty$ and $-\infty$, as follows:\vspace{-0.25cm}
	\begin{align}
	L_+(z_n)&=\lim\limits_{z_n \rightarrow \infty}f_k(z_n)=z_n+\sum_{i=1}^{n-1}z_n\sqrt{\frac{v_i}{v_n}}\nonumber\\
	&=z_n\left(1+\sum_{i=1}^{n-1}\sqrt{\frac{v_i}{v_n}}\right),\\
	L_-(z_n)&=\lim\limits_{z_n \rightarrow -\infty} f_k(z_n)=z_n+\sum_{i=1}^{n-1}(-z_n)\sqrt{\frac{v_i}{v_n}}\nonumber\\
	&=z_n\left(1-\sum_{i=1}^{n-1}\sqrt{\frac{v_i}{v_n}}\right).
	\end{align}
	
	Moreover, to prove that $f_k(z_n)>L_+(z_n)$ we must have:
	\begin{align}
	z_n+\sum_{i=1}^{n-1}\sqrt{\frac{1}{k^2v_n}\left(k^2z^2_nv_i+v_i-v_n\right)}&>z_n\left(1+\sum_{i=1}^{n-1}\sqrt{\frac{v_i}{v_n}}\right)\nonumber\\
	\Leftrightarrow\sum_{i=1}^{n-1}z_n\sqrt{\frac{v_i}{v_n}+\frac{v_i-v_n}{k^2z_n^2v_n}}&>z_n\sum_{i=1}^{n-1}\sqrt{\frac{v_i}{v_n}},\nonumber
	\end{align}
	which always holds true since $(v_i-v_n)/(k^2z_n^2v_n)\geq 0$ for $i\in \mathcal{M}$; with at least one strict inequality for an $i\in \mathcal{M}$. Hence, $f_k(z_n)>L_+(z_n)$.
	Using a similar approach, we can also prove that $ f_k(z_n)>L_-(z_n) $. From $ f_k(z_n)>L_+(z_n),L_-(z_n)$, we can state that $ f_k(z_n) $ always lies in the upper subspace of the intersection of its asymptotes as shown in Fig.~\ref{fig:fzn}. 
	
	Let us define the solutions of $L_+(z_n)=D$ and $L_-(z_n)=D$ as $\bar{z}_n$ and $\munderbar{z}_n$, respectively; then we have: 
	\begin{align}\label{inequal}
	\bar{z}_n>\bar{z}_n^* \geq z_n^\text{min} \geq \munderbar{z}_n^*>\munderbar{z}_n.
	\end{align}
	These inequalities are illustrated in Fig. \ref{fig:fzn}. 
	
	From Proposition~\ref{convexity}, we have $ z_n^\text{min}<0 $, and, thus, $ \munderbar{z}^*_n<0 $. Hence, $\munderbar{z}^*_n<0$ cannot be an acceptable solution for $ f(z_n)=D $ since it violates the condition in \eqref{Maximaconstraints}. We next investigate  $\bar{z}_n^*$. For it to be a valid solution, we must have $ 0<\bar{z}_n^*<R^a$ and the resulting $z_i^*$, which can be computed from~(\ref{solution1tonminus1}), must satisfy the constraints in~(\ref{Maximaconstraints}). Since  $ \bar{z}_n $ is an upper bound of $\bar{z}_n^*$, we start by studying $ \bar{z}_n $. In this respect, next, we prove that $ \bar{z}_n<R^a $:\vspace{-0.2cm}
	\begin{align}\label{zbarplus}
	L_+(z_n)&=D,\Rightarrow z_n\left(1+\sum_{i=1}^{n-1}\sqrt{\frac{v_i}{v_n}}\right)=D, \nonumber\\
	\bar{z}_n&=\frac{D}{\left(1+\sum_{i=1}^{n-1}\sqrt{\frac{v_i}{v_n}}\right)}.
	\end{align}
	
	From~(\ref{zbarplus}), we can see that  $\bar{z}_n<R^a$ always holds true since $ 1+\sum_{i=1}^{n-1}\sqrt{\frac{v_i}{v_n}}>1 $ and $ 0<D<R^a $.
	Hence, since $\bar{z}_n^*<\bar{z}_n$ and $\bar{z}_n<R^a$, then $\bar{z}_n^*<R^a$. 
	
	To satisfy $\bar{z}_n^*>0 $, we must have $ f_k(0)\leq D, $ which yields:
	\begin{align}
	\sum_{i=1}^{n-1}\sqrt{\frac{v_i-v_n}{v_n}}\leq Dk.\label{Dk1}
	\end{align}
	
	Under the condition on $D$ in~(\ref{Dk1}), we obtain $ 0<\bar{z}_n^*<R^a$. However, for $z_n^*=\bar{z}_n^*$ to be a valid solution, the resulting $z_i^*$, for all $i\in \mathcal{M}$ which can be computed from~(\ref{solution1tonminus1}), must satisfy $ 0<z_i^*\leq R^a $. From Theorem~\ref{theoremlocamax}, we know that $z_i^*>0$. To prove that $ z_i^*\leq R^a $ for all $i\in \mathcal{M}$, the maximum $z_i^*$ for all $i\in \mathcal{M}$ must be less that $R^a$. Consider an index $i$ that corresponds to such maximum $z_i^*$. Since $z^*_i=\sqrt{\frac{1}{k^2v_n}\left(k^2\bar{z}_n^{*2}v_i+v_i-v_n\right)}>\sqrt{\frac{1}{k^2v_n}\left(k^2\bar{z}_n^2v_i+v_i-v_n\right)}$ (due to $\bar{z}_n^*<\bar{z}_n$), we show that $\sqrt{\frac{1}{k^2v_n}\left(k^2\bar{z}_n^2v_i+v_i-v_n\right)}\leq R^a$. Therefore, we must have:	
	\begin{align}\label{SecondlimitDK}
	\sqrt{\bar{z}_n^2\frac{v_i}{v_n}+\frac{v_i-v_n}{v_nk^2}}\leq R^a,\nonumber\\
	\Leftrightarrow\frac{D^2}{\left(1+\sum_{i=1}^{n-1}\sqrt{\frac{v_i}{v_n}}\right)^2}\frac{v_i}{v_n}+\frac{v_i-v_n}{v_nk^2}\leq{R^a}^2,\nonumber \\
	\Leftrightarrow\frac{D^2v_i}{\left(\sum_{i=1}^{n}\sqrt{{v_i}}\right)^2}+\frac{v_i-v_n}{v_nk^2}\leq{R^a}^2.
	\end{align}
	
	Since $ D<(n-1)R^b=\frac{n-1}{n}R^a \Rightarrow R^a>(nD)/(n-1)$, and the maximum value for $ v_i $ is 1, then, by comparing the maximum value of the left-hand side of the inequality in~(\ref{SecondlimitDK}) with the minimum value of the right-hand side, we can find a minimum value for $ k $ as follows:
	\begin{align}
	{D}^2+\frac{1}{v_nk^2}\leq{D}^2(\frac{n}{n-1})^2,\nonumber\\
	Dk\geq \frac{n-1}{\sqrt{v_n(2n-1)}}.\label{Dk2}
	\end{align}
	
	From \eqref{Dk1} and \eqref{Dk2}, we can conclude the proof by stating that the unique maximum of $ \tilde{U}^a(\boldsymbol{z},k) $, for $ Dk\geq\max\left\{\frac{n-1}{\sqrt{v_n(2n-1)}},\sum_{i=1}^{n-1}\sqrt{{\frac{v_i-v_n}{v_n}}}\right\} $, is $\boldsymbol{z}^*$ such that $ z_i^* $ for $i\in \mathcal{M}$ is as defined in \eqref{solution1tonminus1}, and $z_n^*=\bar{z}_n^* $.
\end{proof}
In Theorem~\ref{theoremaccsolution}, we derived a unique local maximum of $\tilde{U}^a(\boldsymbol{z},k)$ for $ Dk\geq\max\left\{\frac{n-1}{\sqrt{v_n(2n-1)}},\sum_{i=1}^{n-1}\sqrt{{\frac{v_i-v_n}{v_n}}}\right\} $. In Theorem~\ref{GlobalMax}, we prove that this local maximum is actually a global maximum of $\tilde{U}^a(\boldsymbol{z},k)$. 

\begin{theorem}\label{GlobalMax}
The unique local maximum $ \boldsymbol{z}^*=[z_1^*,\dots,z_{n-1}^*,z_n^* ]$, where $z_n^*=\bar{z}_n^*$ and $ z^*_i=\sqrt{\frac{1}{k^2v_n}\left(k^2z^2_nv_i+v_i-v_n\right)}$ for $i \in  \mathcal{M}$, is a unique global maximum of  $\tilde{U}^a(\boldsymbol{z},k)$ for $ Dk\geq\max\left\{\frac{n-1}{\sqrt{v_n(2n-1)}},\sum_{i=1}^{n-1}\sqrt{{\frac{v_i-v_n}{v_n}}}\right\} $. 
\end{theorem}
\begin{proof}
Theorem~\ref{theoremaccsolution} showed that the derived local maximum of $\tilde{U}^a(\boldsymbol{z},k)$ for a given limit on $Dk$ is unique. Next, we prove that $\tilde{U}^a(\boldsymbol{z},k)$ has no local minimum.

The proof follows the same logic as that of Theorem~\ref{theoremlocamax} by attempting to find solutions in~(\ref{solutions}) which lead to a positive definite Hessian matrix $\boldsymbol{H}$. In this regard, $\boldsymbol{H}$ is positive definite $\Leftrightarrow$ $\boldsymbol{H_1}$ is negative definite $\Leftrightarrow$ all pivots of $\boldsymbol{H}_2$ are negative. As such, a local minimum must have all $z_i^*<0$ for $i\in \mathcal{M} \setminus \{n-1\}$ obtained from~(\ref{solutions}) and must have $\xi<0$ which is defined in~(\ref{Xi}). To this end, similarly to the four-step case analysis derived in  Theorem~\ref{theoremlocamax} over the signs of $z_n$ and $z_{n-1}^*$, we can prove that for $z_i^*<0$ for $i\in\mathcal{M} \setminus \{n-1\}$ no feasible $z_n$ and $z_{n-1}^*$ lead to $\xi<0$. In fact, it can be derived that $(z_{n-1}^*<0, z_n<0)$ and $(z_{n-1}^*>0, z_n>0)$ lead to non-feasible solutions that violate~(\ref{DConstraint}) while $(z_{n-1}^*>0, z_n<0)$ and $(z_{n-1}^*<0, z_n>0)$ lead to $\xi>0$, and hence, a non-positive definite $\boldsymbol{H}$. 

Therefore, $z_n^*=\bar{z}_n^*$ and $ z^*_i=\sqrt{\frac{1}{k^2v_n}\left(k^2z^2_nv_i+v_i-v_n\right)}$ for $i\in \mathcal{M}$ is a unique local maximum of $\tilde{U}^a(\boldsymbol{z},k)$ and $\tilde{U}^a(\boldsymbol{z},k)$ admits no local minima. Hence, this unique local maximum is a global maximum. 
\end{proof}

Based on the derived maximum of $ \tilde{U}^a(\boldsymbol{z},k) $ with respect to $ \boldsymbol{z} $, we next characterize the pure-strategy NE for the GCBG.

\begin{theorem}\label{NEtheorem}
	\hspace{-1mm}For $ k\hspace{-1mm}\geq\hspace{-1mm}\max\left\{\hspace{-0.5mm}\frac{1}{D}\frac{n-1}{\sqrt{v_n(2n-1)}},\frac{1}{D}\sum_{i=1}^{n-1}\sqrt{{\frac{v_i-v_n}{v_n}}}\hspace{-0.5mm}\right\}$, the pure-strategy NE for the GCB game is the allocation vectors defined as follows:
	\begin{align}
		{\boldsymbol{r}^b}^*&=[{r_1^b}^*,\dots,{r_n^b}^*],\label{aNE}\\
		{\boldsymbol{r}^a}^*&={\boldsymbol{r}^b}^*+[z_1^*,\dots,\bar{z}_n^*],\label{bNE}
	\end{align}
	where $ z_n^*=\bar{z}_n^* $ (the positive solution of $ f_k(z_n)=D $), and 
	\begin{align}\label{zis}
		z_i^*= \sqrt{{z_n^{*2}}\frac{v_i}{v_n}+\frac{v_i-v_n}{v_nk^2}},\,\,\,\textrm{for } i\in \mathcal{M}
	\end{align}
	and
	\begin{align}\label{ResourcesrbNE}
	{r_1^b}^*+\dots+{r_n^b}^*=R^b.
	\end{align}
\end{theorem}
\begin{proof}
	As stated in \eqref{maximin}, the pure-strategy NE of the GCBG is the set of two vectors $ \boldsymbol{r}^{a*} $ and  $ \boldsymbol{r}^{b*} $, which solve: 
	\begin{align}
	\min_{\boldsymbol{r}^b\in \mathcal{Q}^b}\max_{\boldsymbol{r}^a\in \mathcal{Q}^a}\tilde{U}^a (\boldsymbol{r}^a,\boldsymbol{r}^b,k).
	\end{align}
	
	Maximizing $ \tilde{U}^a(\boldsymbol{r}^a,\boldsymbol{r}^b,k) $, with respect to $ \boldsymbol{r}^a $ is similar to maximizing $ \tilde{U}^a(\boldsymbol{z},k) $ with respect to $ \boldsymbol{z} $, where $ \boldsymbol{z}= \boldsymbol{r}^a-\boldsymbol{r}^b $. In Theorem \ref{GlobalMax}, we have shown that, if $ Dk\geq\max\left\{\frac{n-1}{\sqrt{v_n(2n-1)}},\sum_{i=1}^{n-1}\sqrt{{\frac{v_i-v_n}{v_n}}}\right\} $, then $ \tilde{U}^a(\boldsymbol{z},k) $ has only one maximim at $ \boldsymbol{z}^*=[z_1^*,\dots,z_{n-1}^*,\bar{z}_n^* ]$, where  $ \bar{z}_n^* $ is the positive solution of $ f_k(z_n)=D $, and $ z_i^* $ is as defined in \eqref{solution1tonminus1}. Therefore, $ \boldsymbol{r}^{a*} $ can be written as follows:
	\begin{align}
	\boldsymbol{r}^{a*}=\boldsymbol{r}^{b}+\boldsymbol{z}^*.
	\end{align}
	
	By substituting $ \boldsymbol{r}^{a*} $ in $ \tilde{U}^a (\boldsymbol{r}^a,\boldsymbol{r}^b,k) $, we will have:
	\begin{align}
	\tilde{U}^a (\boldsymbol{r}^{a*},\boldsymbol{r}^b,k)&=\sum_{i=1}^{n} \left(\frac{v_i}{\pi}\arctan\left(k(r^{a*}_i-r^b_i)\right)+\frac{v_i}{2}\right)\nonumber\\
	&=\sum_{i=1}^{n} \frac{v_i}{\pi}\arctan\left(k(z_i^*)\right)+\frac{1}{2},
	\end{align}
	which has a constant value. Therefore, the minimizing vector $ \boldsymbol{r}^{b*} $ consists of all the strategies in player $b$'s strategy set, $ \mathcal{Q}^b $. Therefore, the pure strategy NE of the GCB is as defined in \eqref{aNE}-\eqref{ResourcesrbNE}. Moreover, the value of the game for each player can be expressed as follows:\vspace{-0.2cm}
	\begin{align}
	V^a(k)\triangleq 1- V^b(k)=\frac{1}{2}+\sum_{i=1}^{n} \frac{v_i}{\pi}\arctan\left(kz_i^*\right).
	\end{align}\vspace{-0.4cm}
\end{proof}

From Theorem~\ref{NEtheorem}, we can see that the GCB admits an infinite number of pure-strategy NEs. However, we have characterized all of them and we have also shown that, by deviating from one NE to the other, the payoff of each player does not change. 

Moreover, Theorem~\ref{NEtheorem} shows that the player with the higher available resources can always achieve a higher payoff by playing the derived pure-strategy NE, since $ z_i^{*}>0 $ for $i\in\mathcal{N}$ which yields $ V^a(k)>V^b(k) $. This implies that, if both players are fully rational, the more resourceful player wins the game in aggregate, which captures a missing feature in the CBG. Indeed, in the CBG, even when a player is more resourceful, it cannot find a winning pure strategy. In this respect, based on the $k-$approximation in the proposed GCBG, we can significantly increase the approximation index $k$ so that the GCBG very closely approximates the CBG and still be able to derive a solution to the game in pure strategies. Here, we note that, when $k$ goes to infinity, our derived solution will no longer hold since the first order conditions in~(\ref{FirstDerivative}) will be met for any $z_i$ ($i\in\mathcal{M}$). This is aligned with the CBG which does not admit an NE in pure strategies, and indeed, confirms that when $k$ goes to infinity our proposed GCBG converges to the classical CBG. Next, we illustrate the GCBG and derived pure-strategy NE using an example.\vspace{-1mm}
\section{An Illustrative Example}\vspace{-0.1cm}
In this numerical example, we consider a GCBG with $10$ battlefields and a random valuation for the battlefields (with a randomly chosen $v_n=0.0215$). We first compute the NE of the game for different values of $D$. The results are shown in Fig.~\ref{fig:allocations}. To satisfy the condition in Theorem~\ref{NEtheorem} on $k$ and $D$, we choose $k=50$.  
As can be seen in Fig.~\ref{fig:allocations}, the difference between the allocated resources over each battlefield at the NE increases as $D$ increases. 

In Fig. \ref{fig:payoff}, we derive the players' payoffs for different $D=R^a-R^b$. As expected, Fig. \ref{fig:payoff} shows that as the difference between the available resources increases, the payoff for the most resourceful player increases. Moreover, in Fig. \ref{fig:kcomparison} we derived the NE for different values of the approximation index, $k$, of the GCBG with $ D=5 $. Fig.~\ref{fig:kcomparison} shows that as $k$ increases, the most resourceful player achieves a higher total payoff. This result illustrates the effect of the $k-$approximation in the GCBG. In fact, as can be seen in~(\ref{kpayoff}) and Fig.~\ref{fig:atan}, for lower values of $k$, the effect of the consumed resources is more accounted for in the expression of the utility function which results in a lower total payoff for the winning player. 
\begin{figure}[!t]
\begin{center}
	\includegraphics[width=0.95\columnwidth]{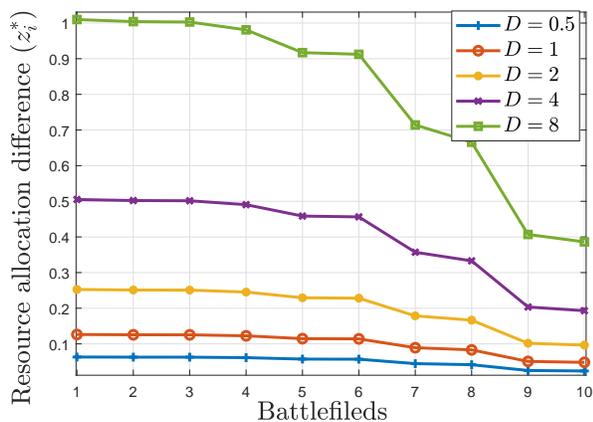}\vspace{-0.2cm}
	\caption{Difference between the allocation of resources at NE}
	\label{fig:allocations}
	\end{center}\vspace{-0.7cm}
\end{figure}
\begin{figure}[!t]
\begin{center}
	\includegraphics[width=0.95\columnwidth]{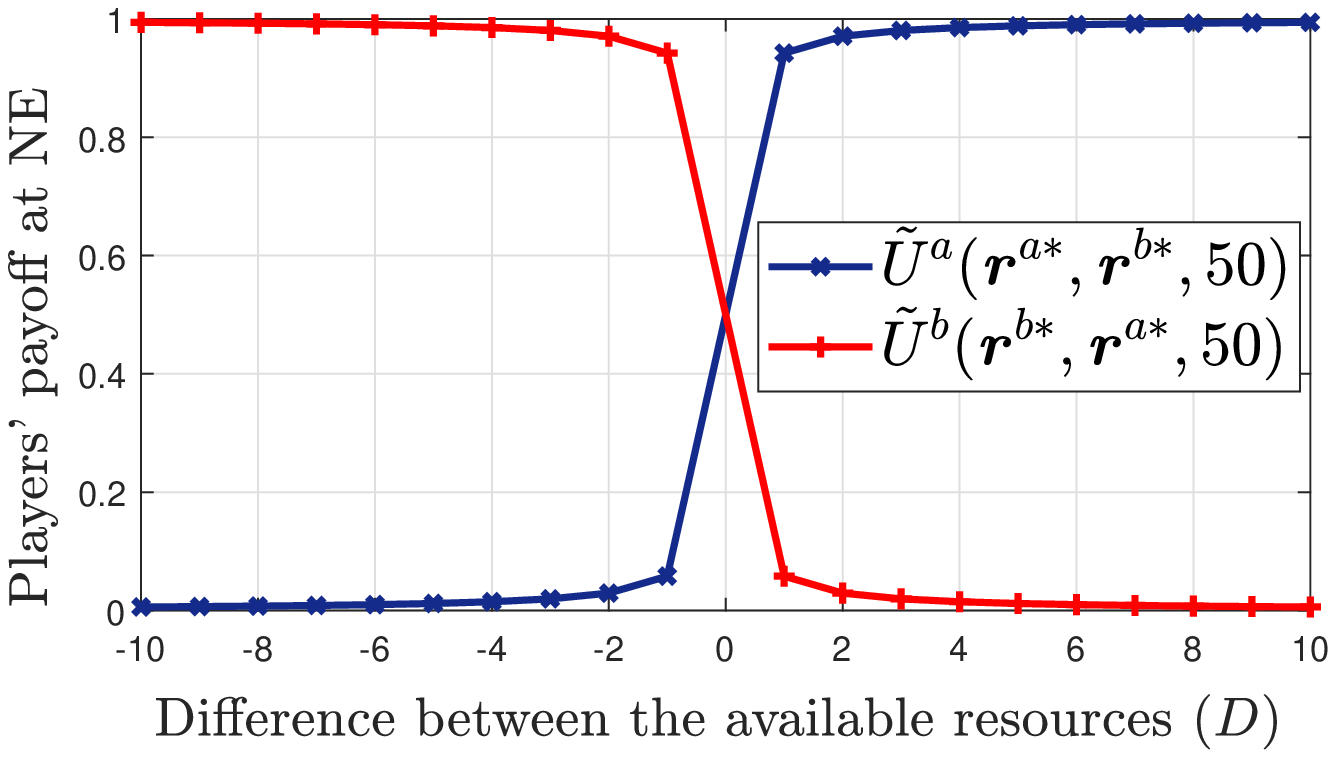}\vspace{-0.2cm}
	\caption{Payoffs of the players at NE.}
	\label{fig:payoff}
	\end{center}\vspace{-0.9cm}
\end{figure}
\begin{figure}[!t]
\begin{center}
	\includegraphics[width=0.95\columnwidth]{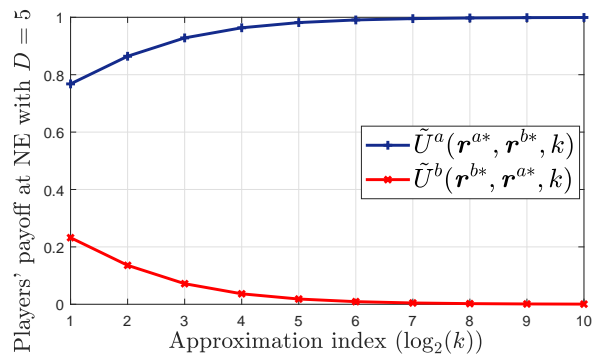}\vspace{-0.2cm}
	\caption{Payoffs of the players at NE with $ D=5 $.}
	\label{fig:kcomparison}
	\end{center}\vspace{-0.8cm}
\end{figure}
\vspace{-2.5mm}
\section{Conclusion}
In this paper, we have proposed a novel generalized Colonel Blotto game framework to analyze competitive resource allocation problems. When the approximation index tends to infinity, we have proven that the proposed game converges to the classical Blotto game. However, in contrast with the Blotto game, our proposed game admits an NE in pure strategies, which renders the derived solutions amenable to practical applications. For this new class of CBG, we have proven the existence of a pure-strategy NE and provided a detailed derivation of closed-form analytical expressions of the NE. Finally, we have presented a numerical example which illustrates the effects of the increase in the number of resources of each player on the game's outcome and NE strategies. Our results also showcased the effect of the approximation index on the solution of the game. 
\def\baselinestretch{0.85}
\bibliographystyle{IEEEtran}
\bibliography{references}
\end{document}